\documentclass[11pt]{article}

\usepackage[a4paper,margin=1.2in]{geometry}

\usepackage{amsfonts,amsmath}
\usepackage{amsmath}
\usepackage{amssymb}
\usepackage{amsthm}
\usepackage{algorithm}
\usepackage{verbatim}
\usepackage{url}
\usepackage{enumitem}
\usepackage{algpseudocode} 
\usepackage{xcolor}

\newtheorem{theorem}{Theorem}
\newtheorem{lemma}{Lemma}
\newtheorem{corollary}{Corollary}

\newcommand{\ket}[1]{\lvert #1 \rangle}

\newcommand{\abs}[1]{\lvert #1 \rvert}
\newcommand{\size}[1]{\lvert #1 \rvert}

\newcommand{\id}{\mathrm{Id}}

\newcommand{\C}{\mathbb{C}}
\DeclareMathOperator{\poly}{\mathrm{poly}}

\DeclareMathOperator{\diag}{diag}

\newcommand{\suppress}[1]{}

\newcommand{\Z}{{\mathbb Z}}
\newcommand{\F}{{\mathbb F}}

\begin{document}

\title{An exact quantum hidden subgroup algorithm and 
applications to solvable groups}

\author{
Muhammad Imran
\\
Department of Algebra, Institute of Mathematics,
\\
Budapest University of Technology and Economics,
\\
M\H{u}egyetem rkp. 3.,
H-1111 Budapest, Hungary.
\\
E-mail: \texttt{mimran@math.bme.hu}
\and
G\'abor Ivanyos
\\
Institute for Computer Science and Control,
\\
E\"otv\"os Lor\'and Research Network, 
\\
Kende u. 13-17., H-1111 Budapest, Hungary.
\\
E-mail: \texttt{Gabor.Ivanyos@sztaki.hu}
}

\maketitle

\begin{abstract}
We present a polynomial time exact quantum  algorithm
for the hidden subgroup problem in $\Z_{m^k}^n$. The
algorithm uses the quantum Fourier transform modulo $m$
and does not require factorization of $m$. For smooth
$m$, i.e., when the prime factors of $m$ are of
size $(\log m)^{O(1)}$, the quantum Fourier
transform can be exactly computed using the method
discovered independently by Cleve and Coppersmith, 
while for general $m$, the algorithm of Mosca and Zalka
is available. Even for $m=3$ and $k=1$ our result appears to be new. 
We also present applications to compute the structure
of abelian and solvable groups whose order has the same (but possibly
unknown)  prime factors as $m$. The applications for solvable groups
also rely on an exact version
of a technique proposed by Watrous for computing
the uniform superposition of elements of subgroups.
\end{abstract}

\section{Introduction}\label{sec:intro}

The first quantum algorithms 
that demonstrate
the advantage of quantum computers over classical ones,
such as the procedure proposed 
by Deutsch and Jozsa~\cite{DeutschJozsa},
are of exact nature, that is, they give the correct
solution with probability one. The speedup of
these though 
is
rather modest. Later 
Bernstein and Vazirani provided an oracle problem that can be 
solved in exact quantum polynomial time, but not in time 
$n^{O(\log n)}$ on any classical bounded-error probabilistic 
algorithm \cite{BerVaz}. Simon's problem  is a simpler
one, and it cannot be solved classically in sub-exponential time~\cite{Simon}.
However, Simon's polynomial time algorithm may fail
with a certain probability
(that can though be made arbitrarily small).
Shor's celebrated polynomial
time quantum algorithms for factoring integers and computing
discrete logarithms build on a generalization of the
Fourier sampling technique used by Bernstein and Vazirani, and Simon
also may fail with a small probability. As observed
by several authors, perhaps first by Kitaev~\cite{Kitaev}, 
Simon's problem and the main ingredients of Shor's  algorithms can 
be cast as instances of the {\em hidden subgroup problem} (HSP).
This is computing a subgroup from a function that is constant on the
cosets of the subgroup and takes distinct values on different cosets,
see the discussion of the main result for a formal definition.
(Note however that Kitaev's {\em stabilizer problem} is a slightly
less general framework).

As already mentioned, exact quantum algorithms return a
correct answer with certainty, while bounded-error methods may
give an incorrect answer with probability bounded by a constant less than $1/2$,
say $1/4$. Zero-error algorithms (also called Las Vegas type algorithms)
such as Simon's and Shor's methods, are required to either give the 
correct answer or report failure with probability at most, say, $3/4$.
Using independent repetition, the error or failure probability can be made 
arbitrarily small. 
Thus a polynomial time  Las Vegas type quantum algorithm can 
be interpreted as a correct algorithm that runs in polynomial time
in the probabilistic sense and in the unlucky case may take a very long time
to succeed, even may (with probability zero) never terminate.
In this respect, the relationship between exact
and "probabilistic" (bounded-error or zero-error) quantum algorithms is 
analogous to that between deterministic and randomized classical algorithms.
Like for the classical counterpart, investigating the power and limitations of exact 
quantum algorithms is a challenging problem. Even for 
various specific algorithmic questions, it would be interesting to know whether exact quantum
algorithms have similar advantage over classical ones as probabilistic
quantum methods. In this paper we address this question for the
the hidden subgroup problem in the additive group $\Z_m^n$ and give
an affirmative answer.

To our knowledge the only "interesting" infinite class of abelian 
groups for which
exact polynomial time quantum algorithms for the general hidden
subgroup problem have been known is $\Z_2^n$. (We consider $G$
interesting if no classical deterministic polynomial time algorithm is known
for the hidden subgroup problem in $G$. When the elementary abelian
subgroups of $G$ are either cyclic or of polynomial size and
the factorization of the order of $G$ is known 
then there is such a classical algorithm: a recursion based on 
determining the intersection of the hidden subgroup with the maximal elementary
abelian subgroups. On the other hand, one can adapt Simon's 
argument~\cite{Simon} to show
that the classical randomized query complexity of the HSP
in $\Z_p^n$ is $\Omega(p^{\frac{n-1}{2}}$.)
It may be worth mentioning that Ettinger, Hoyer and Knill~\cite{EHK04} gave an exact
quantum algorithm for the HSP in an arbitrary finite group that uses
a polynomial number of queries to the function hiding the subgroup. However,
the running time of that algorithm is exponential in general.
Note that the hidden
subgroup problem in the group $\Z_2^n$ includes Simon's problem
as a special case. The first exact
polynomial time hidden subgroup algorithm in $\Z_2^n$
is due to Brassard and Hoyer~\cite{BraHoy}. 
Their method combines Simon's algorithm with a version
of Grover's technique~\cite{Grover} for quantum search.
Later Mihara and Sung
proposed~\cite{MihSung} a simplified version of the Brassard--Hoyer algorithm
for the HSP in $\Z_2^n$ based on a slightly different approach.
Recently, Bonnetain in~\cite{BonnetainSimon}
devised an improved variant that follows a similar structure.
Mosca and Zalka~\cite{MosZal} proposed an exact polynomial
time solution to the discrete logarithm problem in a cyclic
group of known, but not necessarily smooth order $m$. This problem can be cast as
a special instance of the hidden subgroup problem in $\Z_m\times \Z_m$.
The method is based on the exact quantum
Fourier transform algorithm modulo $m$ proposed in the same paper.

In~\cite{BraHoy} the authors propose an approach 
to extend their result to other abelian
groups $G$, e.g., to $G=\Z_m^n$, provided that two problems
become efficiently solved.
One of them
is just the exact computation of the
quantum Fourier transform modulo $m$. 
At the time of publishing \cite{BraHoy},
this was known for smooth $m$, see the manuscripts
of Cleve \cite{Cleve} and Coppersmith \cite{Coppersmith}.
The exact quantum Fourier transform proposed by Mosca and
Zalka~\cite{MosZal}
solves this problem for general $m$. The other
problem is
 constructing a Boolean function on $G$ with
certain properties that support amplitude amplification to
obtain new information about the hidden subgroup with certainty. 
To our knowledge, apart from the case $m=2$, no such construction 
has been published although the authors of~\cite{BraHoy} were 
planning to investigate the issue for groups of smooth order.
In this paper, using ideas from the discrete
logarithm algorithm of Mosca and Zalka, we actually solve a
modified version of the problem for general $m$: we construct 
a function with the desired properties, though on 
$G\times \{0,1\}$, rather than on $G$.

See the remarks at the end 
of Subsection~\ref{subsec:HSP_basic} for details.

\paragraph*{Models for exact quantum computations.}
Bernstein and Vazirani define complexity classes
related to exact quantum computing 
(such the class EQP)
based on quantum Turing machines 
in \cite{BerVaz}. 
Unfortunately, in contrast to bounded error quantum computing,
the power of this model turns to be quite restricted for 
the purposes of exact versions of Fourier sampling based approaches. 
Indeed, as it is 
pointed out by Nishimura and Ozawa in \cite{NishOzaQTM}, 
the quantum Fourier transform $QFT_m$ modulo $m$ can be implemented by a 
single Turing machine only for a {\em finite} set of integers $m$.
Therefore, we consider more powerful models based on {\em quantum circuits}. 

We refer the reader to the paper \cite{NishOzaUnif}
by Nishimura and Ozawa for the details of the definition
of uniform quantum circuit families. Here we provide a brief
summary. A quantum circuit family
is a sequence of circuits built from a given set of elementary 
gates, called the basis. Uniformity means that the circuit 
$C_s$, used for input length $s$, is required to be 
constructible
by a deterministic Turing machine 
in time polynomial in the size of $C_s$ including also that
the numerical parameters of the gates should be computable by
a classical deterministic algorithm in time polynomial in
the circuit size and a prescribed precision. Uniform quantum circuit families
with finite bases are equivalent to quantum Turing machines at least
for polynomial time computations, see \cite{NishOzaPerfect}. 
Cleve \cite{Cleve} and Coppersmith \cite{Coppersmith} noticed
that the quantum Fourier transform $QFT_m$ of $\Z_m$ for smooth $m$
(that is, when all the prime factors of $m$ are polynomially small)
can be exactly computed by an infinitely based
 uniform circuit family of polynomial size. 
Later, Mosca and Zalka \cite{MosZal} proposed a method for
implementing $QFT_m$ exactly by a circuit family 
of size polynomial in $\log m$. That family is uniform in
the sense that the parameters of some gates in the circuit for 
$QFT_m$ can be efficiently (in time $\poly(\log m)$)
built knowing $m$ rather than the input size which is essentially 
$\log m$. This model is apparently stronger than the uniform circuit 
family in the sense of Nishimura and Ozawa. However, in some
applications when $m$ is known to be a member of a set that
can be computed in time polynomial from the {\em size} 
of the input for the specific application, the Mosca-Zalka
implementation of $QFT_m$ can be used as an ingredient of 
a uniform circuit family in the strict sense. 
Besides smooth
numbers, these cases include numbers $m$ 
such as the order of 
the multiplicative group of a finite field
of small characteristic or, more generally,
the order of the general linear group
over such fields. 

We remark that Mosca \cite{MoscaMSRI} proposed an even more powerful model
in which it is  allowed to apply a gate 
$U_\alpha$ with parameter $\alpha$ when (a description of) $\alpha$
emerges {\em during} the computation. (To be more specific,
Mosca proposed allowing gates mapping $\ket{\alpha}\ket{x}$
to $\ket{\alpha}U_\alpha\ket{x}$.) Using that model, he successfully 
"derandomized" Shor's factoring algorithm. The 
square-free decomposition algorithm of
Li et al.~\cite{LPDSsqfree} 
also appears to require application
of certain gates with parameters that depend on the divisors
of the number $N$ to factor and hence are not known to be
classically computable from the input in polynomial time. (To be more specific,
they use $QFT_{d-1}$ to create the uniform superposition of 
integers between $1$ and $d-1$ where $d$ is a "random" divisor 
of $N$ popping up during the computation. Approaches using
amplitude amplification directly could also do the job, 
but all we are aware of need one or more gates with
parameters depending on $d$.) For the problems
considered in this paper, issues with emerging divisors
can be circumvented, therefore we only consider
the three models discussed above. In increasing order of
power, these are 
\begin{enumerate}
\item[(i)] finitely based uniform circuit families; 
\item[(ii)] infinitely based uniform circuit families; and
\item[(iii)] circuit families uniform in the sense that the parameters
of the gates are computed from the input.
\end{enumerate}
In (i) and (ii), uniformity is understood in the strict sense,
that is, parameters of the gates are computed from the input size. 

By a polynomial time 
quantum algorithm with oracles 
we mean a family of polynomial size 
circuits with oracles which,
apart from the oracles, is uniform in the strictest
sense like model (i), that is, 
the input size parametrizes the circuits, 
and the gates are from a finite set. When we do not specify
the number of invocations of the oracles in a statement,  
the number of calls is included in the bound for the 
circuit size. We state our results in a form
when the quantum Fourier transform is accessed via an oracle. This
makes it easier to derive the different implications
in the three models discussed above. Our algorithms will also
be purely unitary, that is, no measurements will be used. The output
will be a state in a separate register and will not be entangled with
the by-products of the computation. In most cases, the output
will be even classical. By that we mean that it is a computational
basis state corresponding to a binary string.

\paragraph*{The main result.}
In this paper, we present a polynomial time exact quantum algorithm
for the hidden subgroup problem in the additive group 
$\Z_{m^k}^n=\Z^n/m^k\Z^n$. Our method combines a
generalization of the Mihara--Sung variant of the Brassard--Hoyer 
algorithm (see also Bonnetain's version) with a solution of a variant
of the key problem in the approach proposed  by Brassard and Hoyer.

The solution is based on some ideas from the 
discrete logarithm procedure of Mosca and Zalka. 

Focusing on groups $\Z_{m^k}^n$ is not too restrictive,
as any abelian group of order having the same prime factors as $m$
is a factor of groups of this type. 
We will consider a slightly generalized version of the 
HSP which can be dealt with the same Fourier sampling based
technique. This generalization captures the application of
the {\em swap test} of Buhrman et al.~\cite{BCWW}
to testing equality of two quantum states 
provided that they are either equal or orthogonal. It
will also be useful for computations in groups. 
Watrous in~\cite{Watrous} also uses a similar generalization 
of Shor's order finding to computing orders of group elements
modulo normal subgroups.
Here is the definition of the generalized HSP.

Let
$f:\Z_{m^k}^n\rightarrow S(\C^{s})$ 
be a "quantum state valued" function (here
$S(\C^s)$ is the unit sphere in $\C^s$). To indicate that
the values of $f$ are {\em states},
we use the notation $\ket{f(x)}$. 
We say that $f$ hides
the subgroup $H\leq \Z_{m^k}^n$ if $\ket{f(x)}=\ket{f(y)}$ 
whenever $x+H=y+H$ while
$\ket{f(x)}$ and $\ket{f(y)}$ are orthogonal if $x+H\neq y+H$. We assume
that $f$ is given by an oracle $U_f$ that maps $\ket{x}\ket{0}$
to $\ket{x}\ket{f(x)}$. Here $\ket{0}$ is a fixed unit vector
in $\C^r$. (If $s=2^t$ then the state with all zero qubits is
a natural choice.) The task is to compute the hidden subgroup
$H$. 
Note
that in the original (and usual) definition of the hidden subgroup
problem, the values of $f$ are from a "classical"
set of $s$ elements. From an instance of that, by taking a bijection between the set
and a basis of $\C^s$ one easily obtains 
an instance of the generalized version. In this paper, we omit the adjective "generalized" and use the term
HSP for the generalized version.
Our main result is the following.

\begin{theorem}
\label{thm:HSP}
There is an exact quantum algorithm that solves
the hidden subgroup problem in $\Z_{m^k}^n$ in time
$(nk\log m)^{O(1)}$ using
$O(nk\log^2 m)$ calls to the oracle $U_f$ or its inverse, and
$O(n^2k\log^2 m)$ applications the quantum Fourier transform
$QFT_m$ of $\Z_m$ or its inverse. 
The procedure outputs the description
of the hidden subgroup $H$ by the  
matrix in Hermite normal form whose columns
are a basis of the lattice which is the inverse image of 
$H$ at the projection map $\Z^n\rightarrow \Z^n/{m^k}\Z^n$. 
If $m$ is a prime then the number of calls to the oracle $U_f$ or its inverse is
$O(nk)$, while the oracle for $QFT_m$ or its inverse is called $O(n^2k)$
times.
\end{theorem}

The numbers $m$, $n$ and $k$ are assumed to be explicitly given as input,
the latter two are in unary. Of course, the most relevant input is provided 
by the oracle $U_f$. 
A unique description of the output is helpful for avoiding
that it gets entangled with the garbage. The Hermite normal
form appears to be a good choice, as it
will also be useful in some of our
applications of the hidden subgroup algorithm.
For a constant $m$, e.g., the product of the first $100$
primes, the theorem gives an efficient exact solution of the HSP
in the strictest model (i) 
(equivalent to quantum Turing machines).
For a smooth $m$, e.g., when it is input in unary, or it
is a product of primes of size polynomial in $n$ and $k$,
the theorem could be used in model (ii), that is 
in the context of uniform
circuits based on an infinite set of gates. For general $m$, through the
implementation of Mosca and Zalka~\cite{MosZal}, one obtains
a result in the strongest model (iii).
For $m=2^\mu-1$, input by $\mu$ in unary, or more generally,
when $m$ can be computed in deterministic polynomial time
from $[\log m]$ (counted in unary), 
application
of \cite{MosZal} again
results in an efficient hidden subgroup algorithm in the
(infinitely based) uniform circuit model (ii).

Koiran, Nesme and Portier showed in~\cite{KoNePo07} that
the bounded-error quantum query complexity of the hidden subgroup
problem in a finite abelian group of rank $n$ (such as the group $\Z_{m^k}^n$) 
is $\Theta(n)$. Thus, regarding the number
of queries, the overhead of our exact method for $\Z_m^n$ is constant
factor if $m$ is a prime number and $O(\log^2 m)$ in the general case. For $\Z_{m^k}^n$
the overhead is $O(k\log^2m)$ (or $O(k)$ when $m$ is prime). In the classical setting,
Nayak~\cite{Nayak21} proposed a simple deterministic polynomial
time method that solves the hidden subgroup problem in an abelian group $G$
using $O(\sqrt{\size{G}})$ queries. The method has an extension to the
non-commutative case that uses $O(\sqrt{\size{G}\log \size{G}})$ queries.
Note that for $G=\Z_2^n$, the classical randomized query complexity of 
the HSP is shown to be
$\Omega(\sqrt{\size{G}})$ by Simon~\cite{Simon}, so in this group Nayak's algorithm has
a constant overhead only.

\paragraph*{Applications.}
Watrous~\cite{Watrous} gave a polynomial time
quantum algorithm for computing the order of
a solvable black-box group with unique encoding of elements.
Black-box groups are used to obtain
general algorithms for various "real" groups such as the multiplicative
groups modulo given integers, permutation groups and matrix groups over finite fields. A matrix group
over the 
$q$-element 
field $\F_q$ is just a subgroup of the group 
$GL_n(\F_q)$ of invertible $n$ by $n$ matrices over $\F_q$. 
Elements of a black-box group are represented by binary strings
of a certain length $\ell$ and the group operations 
are given by oracles and as input, a generating set for the group is given.
See Section~\ref{sec:prelim}
for more details.

We use the exact Fourier transform of $\Z_m$ and the
hidden subgroup algorithm of $\Z_{m}^n$ (for various 
values of $n$) to devise an exact version of Watrous's algorithm
and to solve some other problems for abelian and solvable 
black-box groups with unique encoding of elements
whose order is a divisor of $m^k$ for some $k$.
Our assumption on having the information $m$ about the order
is not standard. In fact, knowing a multiple of the order
of the multiplicative group modulo the number $N$ would
make it possible to factor $N$ in randomized polynomial time. However,
there are situations when a multiple of the order of a group
can be efficiently computed even in deterministic polynomial
time. This is the case for permutation groups and
linear groups over finite fields. Watrous uses a classical Monte Carlo
method by Babai et al.~\cite{BCFLS} to construct a series of subgroups
that his algorithm builds on. Here we will combine
(an exact version of) Watrous's technique with a rather direct
method~\cite{Luks} that finds abelian normal subgroups in
solvable groups to construct the series from the bottom. 

The following theorem summarizes all the applications of the exact quantum hidden subgroup algorithm and the exact version of Watrous's algorithm we obtain.

\begin{theorem}
\label{thm:groups}
There are exact quantum algorithms running in time
$(r\ell\log m)^{O(1)}$ using calls to oracles
for the group operations, the quantum Fourier transform $QTF_m$ of $\Z_m$
and to the inverse oracles that,
in a black box group $G$ with
unique encoding of elements by $\ell$-bit strings given by a list of
$r$ generators, decide whether $G$ is solvable of order dividing $m^k$ for
some positive integer $k$, and if yes, also 
perform the following tasks. 
\begin{enumerate}
\item[(i)] Compute the uniform superposition of elements of $G$.
\item[(ii)] Decide membership in $G$.
\item[(iii)] Compute the order of $G$.
\item[(iv)] Compute the commutator subgroup $G'$ and the derived series of $G$.
\item[(v)] Given a normal subgroup $N$ of $G$ such that
$G/N$ is abelian, decompose $G/N$ as a direct sum of cyclic groups.
\end{enumerate}
\end{theorem}

The uniform superposition in (i) 
is the element
$\ket{G}=\frac{1}{\sqrt{\size{G}}}\sum_{x\in G}\ket{x}$ of the group
algebra $\C G\subseteq \C^{2^\ell}$.

We remark that Luks in~\cite{Luks} presents classical deterministic
polynomial time algorithms for testing solvability of finite {\em matrix
groups} and for many tasks like (iv) in these groups. However, it
is not known whether such classical
deterministic algorithms exist for {\em black-box}
groups, even if oracles for order finding and membership
test in abelian subgroups are allowed.

Of course, the full machinery of the proof of Theorem~\ref{thm:groups}
is not required for the result in (v) for $N=\{1_G\}$, that is computing 
the structure of $G$
when $G$ is abelian of order having the same prime factors as $m$. In fact,
this is a direct consequence of Theorem~\ref{thm:HSP}. Among others, the
result captures finding orders of group elements, in particular in
 the multiplicative group $\F_q^*$ of the $q$-element
field 
$\F_q$. This can be used to find primitive 
elements in certain fields. (Recall that a primitive element of 
$\F_q$ is a field element of multiplicative order $q-1$.) When $q$ is
a power of the prime $p$, a subset of $\F_q$ containing at least one
 primitive element can be constructed deterministically in time
$\poly(\log q)$ assuming the extended Riemann hypothesis (ERH)
when $q=p$ by the results of Wang and Bach~\cite{Wang,Bach} and
when $q=p^2$ by Shoup~\cite{Shoup}. Shoup~\cite{Shoup} and
Shparlinski~\cite{Shparlinski} also proposed methods that work unconditionally 
in time $p\poly(\log q)$. 
Not much later Perel'muter and Shparlinski~\cite{PerShp} came up with
a construction that requires time $\sqrt{p}\poly(\log q)$,
see also Lemma~7 in~\cite{ShparlinskiAdv} for a more explicit bound.
Combining these with order finding, one
immediately obtains the following.

\begin{corollary}
\label{cor:primitive}
Let $q$ be a power of the prime $p$. Then a primitive element
in $\F_q$ can be found by an exact quantum algorithm using
$QFT_{q-1}$ and its inverse in time $\sqrt{p}\poly(\log q)$.
Furthermore, if $q=p$ or $q=p^2$ then,
assuming ERH, a primitive
element in $\F_q$
can be found by an exact quantum algorithm using
$QFT_{q-1}$ and its inverse in time $\poly(\log q)$.
\end{corollary}

We remark that for $p=\poly(\log q)$, 
the first part of this corollary implies existence 
of an exact polynomial time quantum algorithm in model (ii). 

The structure of the paper is the following. In Section~\ref{sec:prelim}
we recall further definitions, discuss the notation and the terminology
used in this work. Some standard facts and 
techniques are also recalled there. 
Section~\ref{sec:HSP} is devoted to the exact hidden subgroup algorithm,
while the applications (mostly in abelian and solvable groups) are
discussed in Section~\ref{sec:groups}.

\section{Notation, terminology and preliminaries}
\label{sec:prelim}

This section 
is devoted to definitions not 
yet
given in the Introduction,
to introducing 
notation and terminology 
that are less known or
not standard and to discussion of some known facts and techniques 
that serve as ingredients of the algorithms presented in this paper.

It is common to define exact quantum procedures as those involving
measurements that produce the desired result with probability one.
By delaying the measurements until 
the end and omitting the final measurement,
one obtains an equivalent unitary procedure. As already mentioned, 
in this paper we consider the unitary version. With a few exceptions 
such as the quantum Fourier transform,
we require (or ensure) that an exact quantum procedure 
is a circuit that implements 
a unitary transformation that maps
$\ket{x}\ket{0}\ket{0}$ to $\ket{x}\ket{output(x)}\ket{garbage(x)}$
where $\ket{x}$ is the computational basis state corresponding to the
input string $x$ and $\ket{output(x)}$ is the state which is the output
corresponding to the actual input string $x$ while $\ket{garbage(x)}$
contains the by-products of the computations. From the initial state
$\ket{x}\ket{0}\ket{0}$ we will frequently omit the third register
as well as in some cases even the second one. Similarly, during the 
description of procedures some registers not
mentioned previously may pop up, meaning a piece of memory
initialized to the zero string of appropriate size. Also, when the content
of a temporary storage is reset to zero, then it can be left away from the 
rest of the description. The output is often classical. By that we mean
that $\ket{output(x)}$ is a computational basis state corresponding to
the string describing the actual output. In that case, the garbage
can be cleared by standard techniques. 

For standard notions from group theory such as subgroups, normal subgroups,
commutator or derived subgroup, derived series, solvability, etc., we refer 
the reader to the textbooks, e.g., to \cite{Robinson}. A finite solvable group
$G$ has a subnormal series with cyclic factors, that is, a sequence 
of subgroups $\{1\}=G_0<G_1<\ldots<G_h=G$ such that
$G_i\lhd G_{i+1}$ and $G_{i+1}/G_i$ is cyclic for $0\leq i<h$.
We will use the term {\em polycyclic series} for such sequences.
A polycyclic series is usually given by a list of group
elements $g_1,\ldots,g_h$
such that $G_i$ is a generated by $g_1,\ldots,g_i$ ($i=1,\ldots,h$).

The concept of black-box groups was introduced by 
Babai and Szemer\'edi \cite{BabSzem}
for studying the structure of finite matrix groups.
As already mentioned, elements of a black-box group (with unique encoding) 
are represented by binary strings of a certain length 
$\ell$ and the group operations (multiplication,
taking the identity element and taking inverses) are given by
oracles. In the quantum setting, the oracle for multiplication
is assumed to be a unitary transformation of $\C^{2^\ell\cdot 3}$
mapping $\ket{x}\ket{y}\ket{0}$ to $\ket{x}\ket{y}\ket{z}$ 
where $z$ encodes the product of the group elements represented by
$x$ and $y$ whenever they both encode valid group elements. 
As we can easily find a
multiple of the order of a given element, we do not actually need the oracles
for taking the identity element and inverses. (Given $x$, a multiple
of the order can be the smallest number of the form $m^t$ with
$1\leq t\leq \ell$
such that $x^{m^t+1}=x$ and then $1_G=x^{m^t}$ and $x^{-1}=x^{m^t-1}$.)
As input, a generating set for the group is given.
We remark that in computational group theory, to capture factor groups
by certain normal subgroups, it is common to also consider black-box 
groups with non-unique encoding of elements. In that case, 
several code-words can represent the same group element and
a further oracle is included to test equality. As computing
the structure of already an elementary abelian black-box group with
non-unique encoding becomes difficult even on quantum computers,
see~\cite{IJS19}, we do not consider this general concept in this paper.

\subsection{Subgroups of $\Z_{m^k}^n$, lattices  
and normal form matrices}

\label{subsec:HNF-SNF}

We will represent a subgroup $A$ of $\Z_{m^k}^n$ by 
a special basis for the lattice $L_A$ in $\Z^n$ 
which is the inverse image of $A$ at the projection map $\Z^n\rightarrow 
\Z^n/m^k\Z^n=\Z_{m^k}^n$. The lattice $L_A$ contains $m^k\Z^n$ therefore
it has full rank $n$.

Recall that for every $n$ by $s$ integer matrix $M$ there exists
an $s$ by $s$ unimodular matrix (an integer matrix with determinant 
$\pm 1$) $U$ such that $H=MU$ where
\begin{itemize}
    \item[1.] $H$ is lower triangular 
whose full zero columns are located to the right of any other column.
    \item[2.] 
The nonzero diagonal entries of $H$ are positive.
    \item[3.] The off-diagonal elements of $H$ are
non-negative and strictly less then the diagonal element
of their rows. 
\end{itemize}
$H$ is called the {\em Hermite normal form} of $M$. 
The Hermite normal form $H$ is uniquely determined
by the lattice $L_M$ in $\Z^n$ spanned by the columns of $M$ 
therefore its nonzero columns form a well-defined basis of $L_M$.
We will use the $n$ by $n$ matrix $H_A$ in Hermite normal
form to represent the lattice $L_A$. For example, if
$A$ is the trivial subgroup of $\Z_{m^k}^n$ then
$L_A=m^k\Z^n$ and $H_A$ is the $n$ by $n$ scalar
matrix $m^k\cdot\id$. 

Also, there exist unimodular matrices $L$ and $R$ 
of sizes $n$ by $n$ and $s$ by $s$, respectively,
such that the matrix $S=LMR$, called the {\em Smith normal form}
of $M$, whose positive entries $d_1,\ldots,d_r$
are located in the first $r$ columns (equivalently, rows)
with $d_{i}|d_{i+1}$ for $1\leq i<r$. The Smith normal form
is also uniquely determined by $L_M$. 

The unimodular matrices $U$, $R$ and $L$ are referred to
as {\em multipliers}. The right multipliers $U$ and $R$ correspond 
to changing the basis for $L_M$ while the left multiplier
$L$ corresponds to changing the basis for $\Z^n$. The normal form
matrices $H$ and $S$ as well as the multipliers can be computed
in deterministic polynomial time (in the number of bits
representing $M$), see e.g.,~\cite{KanBach}.

Besides providing unique representations of subgroups,
the Hermite normal form can be used to solve certain subtasks
efficiently. For example, we will work with pairs of
subgroups $A,B$ of $\Z_{m^k}^n$ for which we already know
that $A\leq B$ and need to decide
whether they are equal. Then this can be readily tested by comparing
the Hermite normal form matrices $H_A$ and $H_B$ representing these
subgroups. Also, if $A<B$ then the first column of $H_B$ whose
diagonal entry is smaller than the corresponding entry in $H_A$
will be an element of $B\setminus A$. 

The Smith normal form
together with the multipliers can be used to solve 
systems of linear congruences such as the following
efficiently. We consider the standard scalar product 
$(x,y)=x^Ty$ modulo $m$. 
For $A\leq \Z_m^n$ the
subgroup $A^\perp$ consists of the elements $y\in \Z_m^n$
such that $(x,y)=0$ for every $x\in A$. We have 
$L_{A^\perp}=\{z\in \Z^n:H_A^Tz\in m\Z^n\}$. Now if 
$S=LH_AR$ is the Smith normal form of $H_A$ then
substituting $z'=(L^T)^{-1}z$ we obtain
$(L^T)^{-1}L_{A^\perp}=\{z'\in \Z^n:S^Tz'\in m\Z^n\}$. A basis of
this lattice can be easily
obtained, as $S$ is diagonal. A
basis for $L_{A^\perp}$ can be obtained
by multiplying by $L^T$. A further Hermite normal form calculation 
gives the unique
representation of $A^\perp$.

Recall that any abelian group $A$ generated by
$n$ elements can be presented as a factor of 
a free abelian group $\Z^n$. If the generators
$a_1,\ldots,a_n$ are fixed, then elements of
the kernel of the map 
$\phi:(x_1,\ldots,x_n)^T\mapsto \prod a_i^{x_i}$
are the {\em relations} corresponding to the generators.
It is usual to write a relation $(x_1,\ldots,x_n)^T$
in the form $\prod a_i^{x_i}$. A set of 
{\em defining relations} is any
collection that generate the kernel.   
An Hermite normal form matrix for the kernel
is a sort of standard and rather economical list
of defining relations for the fixed set of generators.
Theorem~\ref{thm:HSP} will be applicable in the special
cases when $a_i^{m^k}=1_A$ for $i=1,\ldots,n$. Then
the kernel of $\phi$ contains $m^k\Z^n$ and hence
$\phi$ induces a homomorphism $f$ from $\Z_{m^k}^n$
onto $A$. Let $H$ be the subgroup hidden by $f$. Then
the kernel of $\phi$ is the lattice $L_H$. Therefore,
the output of the algorithm of Theorem~\ref{thm:HSP}
is just the Hermite normal form matrix for the relations.
If $A$ is cyclic and $n=1$ then this $1$ by $1$ matrix
contains the order of $u_1$. More generally, the order
of $A$ is the product of the diagonal entries of the
corresponding Hermite normal form matrix. Perhaps
it is worth mentioning the case when $n=2$ and 
$A$ is the cyclic group generated by $u_2$.
Then the Hermite normal
form matrix will be
$$\begin{pmatrix}
1 & 0\\
o_2-d & o_2
\end{pmatrix},$$
where $o_2$ is the order of $u_2$ and $d$ is
the base-$u_2$ discrete logarithm of $u_1$.

By changing the generators, one may obtain
possibly even better and even more economical
presentations. 
In particular, $A$ is isomorphic to the direct
sum $\Z_{m_1}\oplus\ldots\oplus\Z_{m_{n'}}$ with $n'\leq n$,
$m_i>1$, and $m_{i-1}|m_i$ for $1<i$. That is,
there exist generators $b_1,\ldots,b_{n'}$ for $A$
with "diagonal" defining relations 
$b_1^{m_1},b_2^{m_2},\ldots,b_{n'}^{m_{n'}}$.
Such generators and relations can be found using 
the Smith normal form of an initial matrix
of defining relations. In our special case, the Smith
normal form will be the diagonal matrix
$\diag(1,\ldots,1,m_1,\ldots,m_{n'})$ and
the left multiplier gives expressions for the 
new generators in terms of the original ones.

\subsection{Amplitude amplification from $1/\sqrt{2}$ to $1$}

\label{subsec:BraHoy}

Brassard and Hoyer in~\cite{BraHoy} proposed a method,
based on a technique similar to the rotation used in
Grover's search~\cite{Grover}, 
to get rid of the "undesirable half" of a quantum
state provided that the state can be produced by
a unitary procedure. We state here a version of their
result as follows. 

Assume that there is a unitary procedure $\cal U$ acting on $s$ qubits
that maps the state $\ket{0}$ to $c_0\ket{B}\ket{0}+c_1\ket{A}\ket{1}$,
where 
$\ket{A}$ and $\ket{B}$ are of unit norm and the last register
contains just one qubit. 
Then there is a unitary procedure ${\cal U}'$ that, using
$\cal U$ or its inverse 3 times and $O(s)$ ordinary
gates, maps the state $\ket{0}$ to
$c_0'\ket{B}\ket{0}+c_1'\ket{A}\ket{1}$ for some
complex numbers $c_0'$ and $c_1'$ such that in the case when
$c_0=c_1=\frac{1}{\sqrt{2}}$ then $c_0'=0$
and $\abs{c_1'}=1$, that is, the outcome is 
(up to a phase) $\ket{A}\ket{1}$. Also, 
if $c_1=0$ then $c_1'=0$, i.e., in that case
${\cal U}'$ gives $\ket{B}\ket{0}$ (up to a phase).
The result
can be easily extended to the case when $\cal U$ maps
$\ket{x}\ket{0}$ to 
$\frac{1}{\sqrt{2}}\ket{x}\ket{B_x}\ket{0}+\frac{1}{\sqrt{2}}\ket{x}\ket{A_x}\ket{1}$.
That is, $\cal U$ may have an input that is left intact by $\cal U$.

We will apply this in our hidden subgroup 
algorithm (Section~\ref{sec:HSP}). There we
use it in a form closer to that in~\cite{BraHoy}. 
We assume availability of a procedure ${\cal U}_0$ that produces a state
$\sum_{y\in S}c_y\ket{y}\ket{\gamma_y}$ where $S$ is a set of binary
strings of a certain length, $\ket{y}$ denotes the computational
basis vector corresponding to $y$ and $\ket{\gamma_y}$ is of unit norm.
We also suppose that we have a unitary procedure ${\cal U}_1$ that maps $\ket{y}\ket{0}$
to $\ket{y}\ket{f(y)}$ for some function $f:S\rightarrow \{0,1\}$ such
that $\sum_{y\in S:f(y)=1}|c_y|^2=\frac{1}{2}$.
Then $\cal U$ is just the composition 
$\id\otimes {\cal U}_1\circ {\cal U}_0\otimes \id$
and the "desired half" of the state is 
$\ket{A}\ket{1}={\sqrt 2}\sum_{y:f(y)=1}c_y\ket{y}\ket{\gamma_y}\ket{1}$.
We even apply this in certain cases when $\sum_{y\in S:f(y)=1}|c_y|^2>\frac{1}{2}$.
Then we use the idea of Mosca and Zalka applied in~\cite{MosZal}:
by modifying $f$ we throw away some of the good $y$'s and adjust
the coefficient $c_y$ of some others to obtain a superposition
of "total squared amplitude" $\frac{1}{2}$
consisting only of desirable $y$'s.

\subsection{The quantum Fourier transform and Fourier sampling}

\label{subsec:Fourier-sampling}

The elements of the group $\Z_m$ are represented
by integers between $0$ and $m-1$ (in binary). 
The quantum Fourier transform $QFT_m$ of the group
$\Z_m=\Z/m\Z$ maps the computational basis state $\ket{x}$ 
to the state $\frac{1}{\sqrt m}\sum_{y\in \Z_m}\omega^{xy}\ket{y}$,
where $\omega=e^{\frac{2\pi \iota}{m}}$ and
$(\,,\,)$ stands for the usual scalar product on $\Z_m^n$.
Note that the quantum Fourier transform
of $\Z_m^n$ can be implemented as the tensor product of $n$ copies
of $QFT_m$. 

Fourier sampling involves the first few steps of the
known hidden subgroup algorithms. Since we will need some related notation
and since we apply it for the {\em generalized} version of the HSP, 
and also in the exact variant of Mosca's algorithm,
it will be useful to recall details of that standard procedure in our context. 

The procedure works on two registers. 
The first one is for holding group elements from $\Z_m^n$ 
while the second one
is for the states which are the values of $f$.
The first step uses the exact quantum Fourier transform
of $\Z_m^n$ to map the state  $\ket{0}\ket{0}$ to
$$\frac{1}{m^{n/2}}\sum_{x\in \Z_m^n}\ket{x}\ket{0}.$$
Then an application of the oracle $U_f$ gives
$$\frac{1}{m^{n/2}}\sum_{x\in \Z_m^n}\ket{x}\ket{f(x)}.$$
Finally, another application of the quantum Fourier transform
results in
\begin{equation}
\label{eq:Fourier-ori}
\ket{\Psi}=\frac{1}{m^{n}}\sum_{x,y\in \Z_m^n}
\omega^{(x,y)}\ket{y}\ket{f(x)}.
\end{equation}
Let $C$ be a cross-section of $H$ in $\Z_m^n$.
Then every $x\in \Z_m^n$ can uniquely be written as
$x=x'+c$ with $x'\in H$ and $c\in C$.
It follows that
$$\ket{\Psi}=\frac{1}{m^{n}}\sum_{y\in \Z_m^n}\sum_{c\in C}\sum_{x\in H}
\omega^{(x+c,y)}\ket{y}\ket{f(c)}.$$
For fixed $c\in C$ and $y\in \Z_m^n$,
we have 
\begin{equation*}
\sum_{x\in H}\omega^{(x+c,y)}=
\left\{
\begin{array}{ll}
\omega^{(c,y)}\size{H} & \mbox{if $y\in H^\perp$,}\\
0 & \mbox{otherwise.}
\end{array}
\right.
\end{equation*}
(For a proof, notice that map $x\mapsto \omega^{(x,y)}$
is a character $\chi$ of $H$ which is the trivial character $1$ if
and only if $y\in H^\perp$ and use the orthogonality relation
for $\chi$ and $1$.)
Therefore, using also that 
$\size{H}\cdot\size{H^\perp}=\size{H}\cdot\size{C}=m^n$,
we have
\begin{equation}
\label{eq:Fourier-sampling}
\ket{\Psi}=\frac{1}{\sqrt{\size{H^\perp}}}\sum_{y\in
H^\perp}\ket{y}\ket{\gamma_y},
\end{equation}
where
\begin{equation}
\label{eq:Fourier-gamma}
\ket{\gamma_y}=
\frac{1}{\sqrt{\size{C}}}\sum_{c\in C}\omega^{(c,y)}\ket{f(c)}.
\end{equation}

In our hidden subgroup algorithm,
 only the first register of $\ket{\Psi}$ (that is, $\ket{y}$)
will be used and $\ket{\gamma_y}$ will be considered
as garbage. However, $\ket{\gamma_y}$ plays a crucial role
in Watrous's algorithm and its
exact version, see Subsection~\ref{subsec:superpos}.

\section{The exact hidden subgroup algorithm}
\label{sec:HSP}

In this section we prove Theorem~\ref{thm:HSP}. We begin with
the special case $k=1$ and conclude the proof 
with a reduction from the general case to that.

\subsection{The HSP in $\Z_m^n$}

\label{subsec:HSP_basic}

In this subsection we provide the construction of an 
exact algorithm for finding the hidden subgroup $H$ of 
$\Z_m^n$ which requires $O(n\log^2 m)$ queries.
We denote by ${\cal P}$ the procedure described in
Subsection~\ref{subsec:Fourier-sampling}. Recall  
that ${\cal P}$ maps $\ket{0}\ket{0}$ to the state
$\ket{\Psi}$ given in (\ref{eq:Fourier-sampling}).

We use an iteration to compute $H$. During the iteration,
 we maintain a subgroup $K$ of $H$ as well as
 a subgroup $L$ of $H^\perp$. We use the Hermite normal form 
matrices as described in Subsection~\ref{subsec:HNF-SNF} for representing 
$K$ and $L$.
Initially $K=\{0\}$ and $L=\{0\}$ and in 
each round we
enlarge either $K$ or $L$ until $K$ becomes equal to $L^\perp$.
(Note that the conditions on $K$
and $L$ imply that $K\leq L^\perp$ and that $K=H$ if and only if
$K=L^\perp$.)

If $K<L^\perp$ we choose an element $u\in L^\perp \setminus K$.
Computing $L^\perp$, testing equality of $K$ and $L^\perp$, and in the 
case when $K<L^\perp$ finding an $u$ from the difference can be easily 
done using the Hermite and Smith normal form methods described in
Subsection~\ref{subsec:HNF-SNF}.

The map $\mu_u:\Z_m^n\rightarrow \Z_m$ defined
as $\mu_u(y)=(u,y)$ is a homomorphism from $\Z_m^n$ to $\Z_m$. 
Therefore, the image $\mu_u(H^\perp)$ is a subgroup of $\Z_m$
and for each $a\in \mu_u(H^\perp)$
the number of elements $y\in H^\perp$ such that $\mu_u(y)=a$
is $\frac{\size{H^\perp}}{\size{\mu_u(H^\perp)}}$. 
Notice that $\mu_u(H^\perp)$ is the trivial subgroup of $\Z_m$
if and only if $u\in {H^\perp}^\perp=H$. Our aim is to 
find an element $y\in H^\perp$ such that $(u,y)$ 
is nonzero if (and only if) 
$u\not\in H$. To be more specific, for $u\not\in H$ 
we want to "distill" a superposition of certain 
states $\ket{y}\ket{\gamma_y}$
with $(u,y)\neq 0$. To this end, 
we use
the amplitude amplification technique of Brassard
and Hoyer~\cite{BraHoy}~described in Subsection~\ref{subsec:BraHoy}, 
combined with the idea of Mosca and Zalka~\cite{MosZal} to tailor the 
"total squared amplitude" to $\frac{1}{2}$.

Assume that $u\not \in H$. 
Let $d$ be the smallest positive integer 
such that $d+m\Z\in \mu_u(H^\perp)$. 
Then $d$ is a divisor of $m$ and the nonzero elements of $\mu_u(H^\perp)$ 
are represented by the $\frac{m}{d}-1$ positive integers of the 
form $td$ with $td<m$. Also, 
if $\frac{m}{d}$ is even then the integers of the form 
$td$ such that $m/2\leq td<m$ represent just half 
of the elements of $\mu_u(H^\perp)$. However, if
$\frac{m}{d}$ is odd then to get the desired "half", 
we need to add a further element of $\mu_u(H^\perp)$, 
say $d$, with weight $\frac{1}{2}$. The point is that we do not know $d$. 
However, fortunately, for at least one integer $0\leq j\leq \log_2 m$,
namely for $j=\lceil \log_2 d\rceil$, 
the interval $(0,2^j]$ contains only $d$ and no other multiple of $d$.
(Indeed, if $j-1<\log_2 d\leq j$ then $d\leq 2^j$ and $2d>2^j$.)

Based on the discussions above, we have the following exact quantum algorithm for hidden subgroup problem in $\Z_m^n$: 

\begin{algorithm}[H]
\caption{Exact Quantum Algorithm for HSP in $\Z_m^n$}
\scriptsize
\begin{algorithmic}[1]
\State \textbf{Initialize:} $K \gets \{0\}$ and $L\gets\{0\}$;
\While{$K\neq L^{\perp}$}
    \State Take $u\in L^{\perp}\setminus K$;
    \State{${Found}\gets \mbox{False}$;}
    \For{$j=-1,\ldots, \lfloor \log_2 m\rfloor$}\vspace{0.1cm}
	\State{}\Comment{$f(j,x,b)=
		\left\{\begin{array}{ll}
			1 & \mbox{if $(u,x)\geq\frac{m}{2}$ or
					$b=1$ and $0<(u,x)\leq 2^j$} \\
			0 & \mbox{otherwise}
		\end{array}
		\right.$}
	\State{}\Comment{${\mathcal U}_j:
  		\ket{0}\ket{0}\ket{0}\ket{0}\mapsto 
\ket{\psi}=\frac{1}{\sqrt{2|H^{\perp}|}}\sum
\ket{x}\ket{\gamma_x}\ket{b}\ket{f(j,x,b)}$,\mbox{~~~~~~~~}}
	\State{}\Comment{where the summation goes through all 
		$x\in H^\perp$, $b\in \{0,1\}$.\mbox{~~~~~~}} 
	\State{Apply the amplitude amplified version of ${\cal U}_j$;}
	\State{}\Comment{Obtain
$\ket{\psi'_j}=\sum
c'_{f(j,x,b)}\ket{x}\ket{\gamma_x}\ket{b}\ket{f(j,x,b)}$,\mbox{~~~~~~~~~~~~~~~~~~~~~}			}
	\State{Look at the $\ket{x}$-register;}
	\If{$(u,x)\neq 0$}
		\State{${Found}\gets \mbox{True}$;}
		\State{$L\gets \langle L\cup \{x\}\rangle$;}
        \EndIf
    \EndFor
    \If{${Found}=\mbox{False}$}
		\State{$K\gets \langle K\cup \{u\}\rangle$;}
    \EndIf
\EndWhile
\end{algorithmic}
\end{algorithm}

Each round consists of iterations for $j=-1,\ldots,\lfloor \log_2 m\rfloor$ of the following procedure.
Like Mosca and Zalka~\cite{MosZal}, we attach a one-qubit register and,
in addition to calling the procedure $\cal P$, we also apply
a one-qubit Hadamard transform to get the qubit 
$\frac{1}{\sqrt{2}}(\ket{0}+\ket{1})$. We use the 
Brassard-Hoyer amplitude amplification
(see Subsection~\ref{subsec:BraHoy}),
where we accept $\ket{x}\ket{b}$ (compute the qubit $\ket{1}$
in a further register) iff either $\mu_u(x)\geq m/2$ or 
$b=1$ and $0<\mu_u(x)\leq 2^j$.
The case when $\frac{m}{d}$ is even is covered at least once, when $j=-1$
where the interval $(0,2^j]=(0,\frac{1}{2}]$ does not contain any
integer. While the case when $\frac{m}{d}$ is odd is covered at least once, when $j=\lceil \log_2 d\rceil$. 
If $\mu_u(H^\perp)$
is trivial then for each $j$ we get $x$ with $\mu_u(x)=0$, else
for at least one $j$ we get an $x\in H^\perp$ such that $\mu_u(x)$
is nonzero. In the former case, we have $u\in H\setminus K$ 
and we can replace $K$ with
the subgroup generated by $K$ and $u$. In the latter case, $x\in
H^\perp\setminus L$ and, we can enlarge $L$ by replacing it
with the subgroup generated by $L$ and $x$. 

As in each round before termination, the product of the sizes
of the subgroups $K$ and $L$ is increased
by at least a factor $2$, we need at most 
$\lceil \log_2\size{\Z_m^n}\rceil=\lceil n\log_2 m \rceil$ rounds of iterations. 
The overall number of calls to procedure $\cal P$ or its inverse
and through
those the number
of calls to $U_f$ or its inverse as well as the number
of applications of the Fourier transform of $\Z_m^n$ 
or its inverse is
$O(n\log^2 m)$. Thus, $QFT_m$ or its inverse is applied $O(n^2\log^2 m)$ times. 
The required number of elementary gates is 
$(n\log m)^{O(1)}$.

\paragraph*{Remarks.} (1) When $m$ is an odd prime and $u\not\in H$ then
we have $d=1$ therefore considering the case $j=0$ only is sufficient.
Also, in each step $K$ or $L$ gets enlarged by a factor $p$. Therefore,
for a prime $m$ we have an algorithm using $O(n)$ queries.
\\
(2) The key problem in the approach proposed by Brassard and 
Hoyer~\cite{BraHoy} to extend their method from $\Z_2^m$ to
$\Z_m^n$ is existence and efficient computability of a Boolean
function $\lambda:\Z_m^n\rightarrow \{0,1\}$ 
such that $\lambda$ is identically zero on the subgroup what is denoted by
$L$ in our context while it takes value one on exactly half (or another,
prescribed sufficiently large fraction) of the elements
of the (unknown) $H^\perp$ unless $K=H^\perp$. Using that, one can use
the method of Subsection~\ref{subsec:BraHoy}
to obtain an element of $H^\perp\setminus K$ or conclude $K^\perp=H$
with certainty. Notice that if $u\not\in H$ then the function 
$f(j,\cdot,\cdot):\Z_m^n\times \{0,1\}$ defined in Line 6 of the algorithm 
has analogous properties for the appropriate $j$: it is identically zero on $K\times  \{0,1\}$ while
nonzero on exactly the half of the pairs from $H^\perp\times \{0,1\}$. Thus,
 method for tailoring the "total squared amplitude" can be considered as 
a solution of a modified version of the problem of Brassard and Hoyer.
Note that amplitude amplification using $f(j,\cdot,\cdot)$ may happen
to find an 
element of $H^\perp\setminus L$ not only for the appropriate $j$. The current
pseudo-code adds a new element found by the largest "successful" $j$.

In fact, at this point of the algorithm it cannot be determined in general which one of
the functions $f(j,\cdot,\cdot)$ takes nonzero value exactly on the
half of $H^\perp$.

\subsection{Extending to the HSP in $\Z_{m^k}^n$}

We describe how to reduce the hidden subgroup problem of $\Z_{m^k}^n$
to that of $\Z_m^n$. Assume that $f$ hides the subgroup $H$.
We will use an iteration during which we maintain
a subgroup $H_0$ of $H$. In each round $H_0$ is increased until
we conclude that $H=H_0$. Initially $H_0=\{0\}$. 
It will be convenient to represent $H_0$ by a
matrix in Smith normal form (together with multipliers,
in particular the left multiplier for the basis change of $\Z^n$)
whose columns are a basis for the lattice $L_{H_0}$.

We start each round with computing the subgroup 
$K_0=\{x\in \Z_{m^k}^n:mx\in H_0\}$. This can be
efficiently done using the Smith normal form 
representation of $H_0$. (Indeed, if the Smith 
normal form matrix for $H_0$ is $\diag(d_1,\ldots,d_n)$
the matrix for $K_0$ is
$\diag(\frac{d_1}{\gcd(d_1,m)},\ldots,\frac{d_n}{\gcd(d_n,m)})$.
When $K_0=H_0$ we can stop as in that case 
$H$ cannot be larger than $H_0$. Otherwise, we consider the 
factor group $K=K_0/H_0$. The group $K$ is isomorphic to  
$\bigoplus_{i=1}^n\Z_{\gcd(d_i,m)}$ and one can find
a map $\phi_0:\Z_m^n\rightarrow K_0$ such that
the map $x\mapsto
\phi_0(x)+H_0$ is a homomorphism from $\Z_m^n$ onto $K$.
Using the simultaneous diagonal form for $H_0$ and $K_0$
such a $\phi_0$ can be defined by mapping $(x_1,\ldots,x_m)^T$ to 
$(x_1'\frac{d_1}{\gcd(d_1,m)},\ldots,x_n'\frac{d_n}{\gcd(d_n,m)})^T$,
where $x_i'$ is the least {\em positive} integer congruent with
$x_i$ modulo $\gcd(d_i,m)$. 
The composition $f\circ \phi_0$ defines a function hiding a subgroup
$S$ of $\Z_{m}^n$ such that the subgroup generated at $\phi_0(S)$
and $H_0$ generates $K_0\cap H$. We compute generators for
$S$ using the hidden subgroup algorithm for $\Z_m^n$. The images
of the generators at $\phi_0$ and generators for $H_0$ will generate
$K_0\cap H$. We compute the Smith normal form
matrix representing $K_0\cap H$ and test if this subgroup equals
$H_0$. If yes, we can stop because in that case $H=H_0$. 
Otherwise, we replace $H_0$ with $K_0\cap H$ and proceed 
with the next round.
Using an induction on $j$, one can show that $H_0$ after the $j$th
round contains $H\cap m^{k-j}\Z_{m^k}^n$ ($j\leq k$). 
Therefore, the procedure requires at most $k$ rounds.

\section{Applications to groups}

\label{sec:groups}

In this section we prove Theorem~\ref{thm:groups}. We start with
some rather direct applications of the exact hidden subgroup algorithm
(Subsection~\ref{subsec:direct-appl}). For applicability in the 
proof of Theorem~\ref{thm:groups}, some of the statements are rather technical, 
they include the assumption of availability of a procedure that
computes the uniform superposition
$\ket{K}$ of the element of a subgroup $K$.
Of course, when $K$ is the trivial subgroup $\{1_G\}$ then
this is easy as $\ket{K}=\ket{1_G}$. 
For this $K$ we obtain results
for abelian groups that may be of interest on their own right.
Throughout the section we assume that $G$ is a black box group
with unique encoding of elements with binary strings of
length $\ell$.

\subsection{Some basic applications}
\label{subsec:direct-appl}

\subsubsection{An exact version of the swap test}
\label{subsubsec:swap}

The swap test can be used to decide equality of two quantum states,
provided that they are either equal (up to a phase) or orthogonal 
to each other. In the usual setting, the states are given as input
and orthogonality is detected with probability $\frac{1}{2}$. This
suggests that the amplitude amplification technique of Brassard
and Hoyer (see Subsection~\ref{subsec:BraHoy}) can be used
to get an exact version. However, availability of the two states
as input is not sufficient for applying the technique. It rather
requires {\em procedures for creating} the two states. In that 
setting, one could directly combine the swap test with the amplitude
amplification to obtain an exact procedure. Interestingly, the
idea of swapping the two states conditionally 
suggests an interpretation as an instance 
of (the generalized version of) the hidden subgroup problem in
the two-element group $\Z_2$.
\begin{corollary}
\label{cor:swap}
Assume that there are two unitary procedures ${\cal P}_1$
and ${\cal P}_2$ mapping the $\ell$-qubit state $\ket{0}$
to $\ket{\psi_1}$ and to $\ket{\psi_2}$, respectively. Suppose further
that the states $\ket{\psi_1}$ and $\ket{\psi_2}$ are either
equal (up to a phase) or orthogonal. Then there is an exact
procedure that returns $\ket{0}$ if the two states are orthogonal
and $\ket{1}$ if they are equal. The procedure uses
$O(1)$ applications of ${\cal P}_1$, ${\cal P}_2$ and their inverses
and $O(\ell)$
elementary gates. 
\end{corollary}

\begin{proof}
We consider the quantum state valued function $f$ from $\Z_2$
to $\C^{2\ell}$ such that $\ket{f(0)}=\ket{\psi_1}\ket{\psi_2}$
and $\ket{f(1)}=\ket{\psi_2}\ket{\psi_1}$. Note that
$\ket{f(0)}$ and $\ket{f(1)}$ are either equal 
(if $\ket{\psi_1}=c\ket{\psi_2}$) or orthogonal. Also,
if the two states are equal then the subgroup hidden by $f$
is the entire group $\Z_2$ while in the other case
it is the trivial subgroup. A procedure computing
$\ket{x}\ket{f(x)}$ from $\ket{x}\ket{0}$ can
be built using ${\cal P}_1$
and ${\cal P}_2$ and $\ell$ controlled swap of qubits. 
The proof can be concluded by an application
of Theorem~\ref{thm:HSP}.
\end{proof}

We remark that as the HSP is in $\Z_2$, already
 the hidden subgroup algorithm of Brassard and
Hoyer~\cite{BraHoy} implies the result. 
Also, by unfolding the hidden subgroup algorithm,
the proof would give essentially the same circuit 
as a direct combination of the swap test and the 
amplitude amplification would result in. 

\subsubsection{Testing membership}

The swap test can be used to reduce
Task (ii) of Theorem~\ref{thm:groups} to Task (i).

\begin{corollary}
\label{cor:member}
Let $K$ be a subgroup of $G$ and assume that we have a unitary 
procedure ${\cal S}_K$ for creating the uniform superposition $\ket{K}$
of the elements of $K$. Then we can test membership of $u\in G$ in $K$
in time $\ell^{O(1)}$ by an exact quantum algorithm that uses $O(1)$
applications of the group oracles and ${\cal S}_K$ or the inverses
of these.   
\end{corollary}

In the statement, the subgroup $K$ is hidden 
behind the procedure ${\cal S}_K$. When we apply 
this corollary, ${\cal S}_K$ will be actually 
implemented by a procedure performing Task~(i)
of Theorem~\ref{thm:groups} for $K$ in place of $G$.
That procedure uses (generators for) $K$ as input.

\begin{proof}
We apply Corollary~\ref{cor:swap} with ${\cal P}_1={\cal S}_K$ and
${\cal P}_2=\mu_u\circ {\cal S}_K$ where $\mu_u(v)=uv$.
\end{proof}

\subsubsection{Presentation of an abelian factor}
\label{subsubsec:present}

The following corollary provides an algorithm to efficiently compute 
presentations and decomposition of abelian factors using the exact 
hidden subgroup algorithm.

\begin{corollary}
\label{cor:present}
Let $K$ be a normal subgroup of $G$ and assume that there 
is unitary procedure ${\cal S}_K$ that (on input $\ket{0}$)
computes $\ket{K}$.
Suppose further that
we are given elements $u_1,\ldots,u_n\in G$
such that $u_i^{m^k}=1$ for some integer $k\geq 0$;
$[u_i,u_j]\in K$ ($i,j=1,\ldots,n$);
and $G$ is generated by $u_1,\ldots,u_n$ and $K$. 
Then there are exact quantum algorithms that,
in time $(\ell\log m)^{O(1)}$, compute a presentation
of $G/K$ in terms of the generators $u_1K,\ldots,u_nK$
and an isomorphism $G/K\cong 
\Z_{m_1}\oplus\ldots\oplus \Z_{m_{n'}}$ with
$m_i|m_{i-1}$ for $1<i\leq n'$. The procedures 
apply ${\cal S}_K$, $QFT_m$, the group oracle(s) and the 
inverses of these. The aforementioned isomorphism will be given by listing
$n'$, the sequence $m_1,\ldots,m_{n'}$ and
an $n'$ by $n$ matrix $(\alpha_{ij})$ such that
with $z_i=\prod_{j=1}^n u_j^{\alpha_{ij}}$ ($i,1,\ldots,n'$),
we have $z_i^{m_i}\in K$ but  
$\prod_{i=1}^{n'}z_i^{\beta_i}\in K$ with $0\leq \beta_i<m_i$
implies that all $\beta_i=0$. 
\end{corollary}

\begin{proof}
As the order of $u$ is at most $2^\ell$, there
exist an integer $k$ between $0$ and $\ell$
such that $u^{m^k}=1$.
(Let $\alpha_p$ be the multiplicity of the prime
$p$ in the order of $u$. Then $\alpha_p\leq \ell$
and if $k\geq \alpha_p$ for every prime factor $p$ of $m$,
we have $u^{m^k}=1$.)  We can find such a $k$ using
at most $\ell$ trials. Note that these trials can also be used to decide
if the order of $u$ is a divisor of a power of $m$.

Consider the function $f:\Z_{m^k}^n\rightarrow \C G$ defined as
$$f(x_1,\ldots,x_n)=\Big |\Big(\prod_{i=1}^nu_i^{x_i}\Big)K\Big\rangle.$$ 
The function $f$ hides the subgroup $H$ of $\Z_m^n$ 
such that $L_H$ is the lattice of relations for a presentation of $G/K$
in terms of generators $u_iK$ $(i=1,\ldots,n)$. We apply
Theorem~\ref{thm:HSP} to compute $H$. The result will actually be
an Hermite normal form matrix whose columns are a basis for $L_H$.
The Smith normal form of that matrix and the left multiplier
give the stated isomorphism.
\end{proof}

We remark that Corollary~\ref{cor:present} captures
many tasks, such as finding orders of elements of $G$,
finding generators for cyclic subgroups and solving the
discrete logarithm problem for pairs of elements of $G$.
It could also serve as an alternative method for testing
membership in $K$.

\subsection{Computing the uniform superposition of group elements}
\label{subsec:superpos}

In this part, we show how to perform Task (i) of Theorem~\ref{thm:groups} 
provided that we have a polycyclic series 
$\{1\}=G_0<G_1<\ldots<G_h=G$ given
by elements $g_1,\ldots,g_h$
such that $G_i$ is a generated by $g_1,\ldots,g_i$ ($i=1,\ldots,h$).
We also require that the order of each factor group $G_{i}/G_{i-1}$
is a divisor of $m$. This will be ensured when we construct
the polycyclic series. We use Watrous's method to
build a "pyramid" of superpositions over the subgroups in the series. The
key step is an exact version of Watrous's technique
described in Section 3.2 of \cite{Watrous} that,
from $s$ copies of the uniform superposition $\ket{G_i}$
of elements of $G_i$ computes $s-1$ copies of $\ket{G_{i+1}}$.
Then creating $\ket{G}$ goes as follows. We start with $h+1$
copies of $\ket{1}$ and from these we compute $h+1-i$
copies of $\ket{G_i}$ ($i=1,\ldots,h$).

Below we describe the key step. As a tool,
we use the following. 

\begin{lemma}
\label{lem:gcd}
Let $z_1,\ldots,z_s$ be integers between $0$ and $m-1$. Then
there is a
 deterministic algorithm that 
in time $(s\log m)^{O(1)}$
finds integers $u_1,\ldots,u_{s-1}$ such that
$\gcd(u_1z_1+u_2z_2+\ldots+u_{s-1}z_{s-1}+z_s,m)
=\gcd(z_1,\ldots,z_s,m)$.
\end{lemma}

\begin{proof}
There is an easy reduction to the case $s=2$. Indeed,
having a method for $s=2$ we can find $u_{s-1}$ such that
$\gcd(u_{s-1}z_{s-1}+z_s,m)=\gcd(z_{s-1},z_s,m)$ and then
$\gcd(z_1,\ldots,z_{s-2},u_{s-1}z_{s-1}+z_s,m)=\gcd(z_1,\ldots,z_{s-1},z_s,m)$.
We proceed with finding an appropriate coefficient $u_{s-2}$, and so
on. 

To solve the case $s=2$, given $z_1,z_2$ and $m$, we need
to find $u$ such that $\gcd(uz_1+z_2,m)=\gcd(z_1,z_2,m)$.
Dividing by $\gcd(z_1,z_2,m)$, we obtain the case when 
$\gcd(z_1,z_2,m)=1$ and when we are looking for $u$ such that
$\gcd(uz_1+z_2,m)=1$. For each prime divisor $p$ of $m$
we say that $u$ is "bad" modulo $p$ if $p|uz_1+z_2$, otherwise
$p$ is "good" modulo $p$. Obviously, goodness and badness
depend only on the residue class of $u$ modulo $p$ and
for every $p$ there are $p-1$ "good" residue classes.
Let $S$ be a positive integer to be determined
later and let $P$ be the set of the prime divisors
of $m$ less than $S$. These primes can be listed 
and their product $m'$ can be computed
in time $(S\log m)^{O(1)}$. For each prime in $P$
pick a "good" residue class $u_p$ modulo $p$ and by Chinese
remaindering compute an integer $0\leq u_0<m'$ such that 
$u_0\equiv u_p$ modulo $p$ for every $p\in P$. 
Then for every integer $t$, $u_t=u_0+tm'$ is good modulo every prime in $P$.
We take the sequence $u_0,\ldots,u_{S-1}$ and compute
$\gcd(u_tz_1+z_2,m)$ for $t=0,\ldots,S-1$. Clearly, for each prime
$p>S$, the sequence does not contain two members that
are from the same residue class modulo $p$. Therefore, 
at most one of the $u_t$s is bad modulo $p$. Thus, if $S>\log_2 m$
then there is at least one $t$ such that $u_t$ is good modulo
every prime divisor of $m$. We take that $u_t$ (reduced modulo $m$).
\end{proof}

Let $N$ be a subgroup of $G$ and let $u\in G$ such that
$u^{-1}Nu=N$ and $u^m\in N$. Let $K$ be the subgroup generated by $u$ and $N$.
Then $N\lhd K$.
Assume that we have $s$ copies of the state
$\ket{N}=\frac{1}{\sqrt{\size{N}}}\sum_{v\in N}\ket{v}$.
From these, we shall make $s-1$ copies of the state 
$\ket{K}=\frac{1}{\sqrt{\size{K}}}\sum_{w\in K}\ket{w}$.
We use Watrous's method with two modifications.
The first thereof is that we take the multiple $m$ instead
of the order of $u$ modulo $N$ so that
we do not need to use any Fourier transform other than $QFT_m$.
The second modification is
that, as we want an exact procedure, we have to be prepared
to handle some "degenerate" cases that occur with small
probability. (Watrous discarded these and used repetition.
However, the possibility of a variant that also works in the 
degenerate cases is mentioned in~\cite{Watrous}.)

Notice that $\ket{K}=\frac{1}{\sqrt{m}}\sum_{x\in \Z_m}\ket{u^xN}$.
We apply the Fourier sampling described in 
Subsection~\ref{subsec:Fourier-sampling} for
the function 
$f:\Z_m\rightarrow \C G$ given as $f(x)=\ket{u^xN}$. This time
we are not interested
which subgroup is hidden by $f$, our objectives are the states
$\ket{\gamma_y}$ in the second register of $\ket{\Psi}$
in~(\ref{eq:Fourier-sampling}) and in~(\ref{eq:Fourier-gamma}). 
In this special case, 
$$\ket{\gamma_y}=\frac{1}{\sqrt m}\sum_{x\in \Z_m}\omega^{xy}\ket{u^xN}.$$ 
(Here we go back to~(\ref{eq:Fourier-ori}) and do not use the decomposition
of $\ket{\gamma_y}$ by the cosets of the hidden subgroup.)
Another minor difference from Subsection~\ref{subsec:Fourier-sampling}
is that we compute $f$ in a register initially $\ket{N}$
rather than $\ket{0}$.

By applying Fourier sampling on the available $s$ copies of $\ket{N}$,
we obtain 
$$(m)^{-s/2}\sum_{y_1,\ldots,y_s=0}^{m-1}
\ket{y_1}\ldots\ket{y_s}\ket{\gamma_{y_1}}\ldots\ket{\gamma_{y_s}}.$$

We consider the terms of this sum. 
For each $s$-tuple $(y_1,\ldots,y_s)$
we will change the first $s-1$ states to
$\ket{\gamma_0}=\ket{K}$ using Watrous's trick.
The key
fact behind is that $\ket{\gamma_y}$ is an eigenstate
with eigenvalue $\omega^{-y}$ 
for the action of $u$ by multiplication. Based on this,
it is easy to show that
if we have a pair $\ket{\gamma_y}\ket{\gamma_z}$
of such states then multiplying the content of the
first part by the $t$th power of the content 
of the second part, the effect can be interpreted as
the first part $\ket{\gamma_y}$ remains unchanged while the second part
becomes $\ket{\gamma_{z-ty}}$, where $z-ty$ is understood modulo $m$.
We use this first to arrange that 
the index $y_s$ of the last state $\gamma_{y_s}$ becomes a generator
for the subgroup $L$ of $\Z_m$ generated by $y_1,\ldots,y_s$. To this end, 
using the method of Lemma~\ref{lem:gcd}, we 
find integers $u_1,\ldots,u_{s-1}$ such that
$\gcd(u_1y_1+\ldots+u_{s-1}y_{s-1}+y_s,m)=\gcd(y_1,\ldots,y_s,m)$.
Then $u_1y_1+\ldots+u_{s-1}y_{s-1}+y_s$ modulo $m$ is a generator
for $L$.  We multiply (the content of) 
$\ket{\gamma_{y_1}}$ by the $-u_1$th power of $\ket{\gamma_{y_s}}$,
then $\ket{\gamma_{y_2}}$ the $-u_2$th power of 
by an appropriate power of (the new)
$\ket{\gamma_{y_s}}$, and so on. Eventually $y_s$ becomes
the generator $u_1y_1+\ldots+u_{s-1}y_{s-1}+y_s$ modulo $m$.
Then $y_i\equiv t_iy_s$ modulo $m$ 
for some integer $t_i$ ($i=1,\ldots,s-1$), 
and multiplying
$\ket{\gamma_{y_s}}$ with the $t_i$th power of 
$\ket{\gamma_{y_i}}$ changes $\ket{\gamma_{y_i}}$ to 
$\ket{\gamma_{y_i-t_iy_s}}=\ket{\gamma_{0}}=\ket{K}$.

\subsection{Constructing a polycyclic series}

We start with testing if the orders of the generators
are divisors of some power of $m$.
(The proof of Corollary~\ref{cor:present} includes a straightforward 
classical method doing this.) If one of the tests fails,
then no further computation is needed: we can return that the order
of $G$ does not satisfy the required property. 
When a new element occurs, we can perform this test and 
exit in case of failure. We
build a polycyclic series from the bottom. 

Assume that $N\lhd G$ and we already have a polycyclic series
of $N$ with factors of order dividing $m$. 
Initially, $N$ is the trivial subgroup. By the previous
subsection, we have an efficient method for Task
(i) of Theorem~\ref{thm:groups} with $N$ 
in place of $G$ and one can use Corollary~\ref{cor:member}
(or Corollary~\ref{cor:present})
to test membership in $N$. Our first goal is to find a normal subgroup
$K\lhd G$ such that $K/N$ is abelian. To this end, we use
a black box method from Luks's paper \cite{Luks}.
Let $y_1,\ldots,y_{r'}$ be those of the given generators
for $G$ that are not in $N$. We may assume that $r'\geq 1$ as otherwise
$G=N$. Put $x=y_1$. We collect a list $L$ of conjugates of $x$ 
until we find two elements $u,v\in L$ such
that $[u,v]\not\in N$ or the subgroup $K$ generated by $L$ and
$N$ stabilizes. Initially $L$ contains only $x$. When a new conjugate
$v$ of $x$ is added to the list then we check if $[u,v]\in N$
for every $u$ already in $L$. If an $u$ is found
such that $[u,v]\not\in N$ then we replace
$x$ with $[u,v]$. We test if the order of the new $x$ 
is a divisor of some power of $m$ and stop if not.
Otherwise, we restart collecting conjugates of $x$.

Stabilization can be decided as follows.
Notice that $K/N$ is abelian
and $L$ directly extends the polycyclic series of
$N$ to $K$. The series can also be refined to have
factors of order dividing $m$ in an obvious way.
Based on this, there is an efficient method for Task (i) of 
Theorem~\ref{thm:groups} for $K$ as well, whence we can
test membership in $K$ like in $N$. 
Using that, for every $u$ in the list $L$ we test if $u^{y_i}\in K$
for every $y_i$. If this is the case, then $K\lhd G$.
Otherwise, we can add a new conjugate $z^{y_i}$ of $x$ to $L$.
When $K$ stabilizes, we replace $N$ with $K$ 
and repeat the procedure outlined above. As $\size{L}\leq
\log_2\size{K/N}\leq \ell$, either $K$ stabilizes or two 
conjugates of $x$ is found with commutator not in $N$ in at most 
$\log_2 G\leq \ell$ rounds.

If $G/N$ is solvable and $N<G$ then one eventually finds
a normal subgroup $K$ of $G$ properly containing $N$
because if $x\in G^{(i)}$ for some $i$ and if $u,v$
are conjugates of $x$ then $[u,v]\in [G^{(i)},G^{(i)}]=G^{(i+1)}$.
Therefore, to detect non-solvability, we keep track how many times
the element $x$ has been updated. If it has happened
more than $\ell$ times, then
we can stop and conclude that $G$ is not solvable.

\subsection{The order, testing membership and abelian factors}

The order of $G$ is the product of the orders of the factors
$G_{i+1}/G_i$ of a polycyclic series $\{1\}=G_0<G_1\ldots<G_h=G$. 
The order of a factor can be computed by the algorithm of
Corollary~\ref{cor:present}. Testing membership
in a subgroup $N\leq G$ can be done by 
building a polycyclic series of $N$ and then using
Corollary~\ref{cor:member}. (For testing membership of $x$
in $G$ itself, we replace $N$ with $G$ and $G$ with the group
generated by $x$ and $G$.) Similarly,
decomposing an abelian factor
$G/N$ can be done by constructing first
a polycyclic series of $N$ and then applying 
Corollary~\ref{cor:present}.

\subsection{The derived series}
We first show how to compute the commutator subgroup $G'$.
The next to last element in a polycyclic series is
a normal subgroup $N\lhd G$ such that $G/N$ is cyclic. The data structure
for the polycyclic series algorithm includes an element $g$ such that 
$g$ and $N$ generate $G$.
By recursion, we start with computing the commutator subgroup $N'$. 
As $N'$ is a characteristic subgroup of $N$, we have $N'\lhd G$. 
We can also find a polycyclic series
for $N'$ which enables us computing $\ket{N'}$. We take the iterated
commutators of the generators for $N$ with $g$ until the subgroup
generated by $N$ and these commutators stabilizes. This subgroup
will be $G'$.
To check stabilization,
we use Corollary~\ref{cor:member} to test
membership in the intermediate subgroups.

By recursion, we compute the derived series of $G'$ which extends
to that of $G$ in the obvious way.

\paragraph*{Acknowledgments.} The authors are grateful
to Lajos R\'onyai, to Igor Shparlinski and to an anonymous referee
 for their helpful comments and suggestions.
The research of the second author was supported by the Hungarian Ministry 
of Innovation and Technology NRDI Office within the framework of the Artificial
Intelligence National Laboratory Program.

\bibliographystyle{alpha}
\bibliography{myrefs}

\end{document}